\documentclass[a4paper]{article}

\usepackage{hyperref}
\usepackage{algorithmicx}
\usepackage[utf8]{inputenc}
\usepackage{algorithm}
\usepackage{algpseudocode}
\usepackage[margin=.85in]{geometry}
\usepackage{graphicx}
\usepackage{float}
\usepackage{subfig}
\usepackage{amsthm,color}
\usepackage{amsmath,amssymb}
\usepackage[dvipsnames]{xcolor}
\usepackage{paralist}
\usepackage[title]{appendix}
\usepackage{systeme}
\usepackage{soul}
\usepackage{multirow}
\usepackage{multicol}

\newcounter{casectr}

\newcommand{\rmv}[1]{}

\definecolor{purple}{rgb}{0.6,0,0.6}

\definecolor{darkgreen}{rgb}{0,0.4,0}

\algrenewcommand\algorithmicindent{3.7mm}

\newcounter{theorem}
\newcounter{NewCounter}
\newcounter{claimcount}[NewCounter]
\makeatletter
\renewcommand{\p@claimcount}{\theNewCounter.}
\makeatother

\newcounter{cclaimcount}[claimcount]
\makeatletter
\renewcommand{\p@cclaimcount}{\theNewCounter.\theclaimcount.}
\makeatother

\newtheorem{theorem}[NewCounter]{Theorem}

\newtheorem{corollary}[NewCounter]{Corollary}
\newtheorem{lemma}[NewCounter]{Lemma}
\newtheorem{observation}[NewCounter]{Observation}

\newtheorem{remark}[theorem]{Remark}

\newtheorem{definition}[NewCounter]{Definition}

\title{Faster Algorithms for Bounded-Difference Min-Plus Product}

\author{Shucheng Chi \thanks{chishuchengggg@icloud.com} \\ Tsinghua University
\and Ran Duan \thanks{duanran@mail.tsinghua.edu.cn} \\ Tsinghua University
\and Tianle Xie \thanks{xtl17@tsinghua.org.cn} \\ Tsinghua University} 

\begin{document}
	\maketitle
	
	\begin{abstract}
		Min-plus product of two $n\times n$ matrices is a fundamental problem in algorithm research. It is known to be equivalent to APSP, and in general it has no truly subcubic algorithms. In this paper, we focus on the min-plus product on a special class of matrices, called $\delta$-bounded-difference matrices, in which the difference between any two adjacent entries is bounded by $\delta=O(1)$. Our algorithm runs in randomized time $O(n^{2.779})$ by the fast rectangular matrix multiplication algorithm [Le Gall \& Urrutia 18], better than $\tilde{O}(n^{2+\omega/3})=O(n^{2.791})$ ($\omega<2.373$ [Alman \& V.V.Williams 20]). This improves previous result of $\tilde{O}(n^{2.824})$ [Bringmann et al. 16]. When $\omega=2$ in the ideal case, our complexity is $\tilde{O}(n^{2+2/3})$, improving Bringmann et al.'s result of $\tilde{O}(n^{2.755})$.
	\end{abstract}
	
	\newpage
	
	\section{Introduction}\label{section-introduction}

The min-plus product $C$ of two matrices $A$ and $B$ with sizes $a\times b$ and $b\times c$, denoted by $C=A\star B,$ is an $a\times c$ matrix where $C_{i,j}=\min_{1\le k\le b}\left\{A_{i,k}+B_{k,j}\right\}.$
Solving the min-plus product of two $n\times n$ square matrices is asymptotically equivalent to all pair shortest path (APSP) problem. Finding faster algorithms for APSP (so for min-plus product) has been a major problem for computer scientists. 
Recent breakthroughs include the one by Williams \cite{2018Faster}, which works in time $O(n^3/c^{\sqrt{\log{n}}}).$ However, it still remains open whether there exists a truly subcubic algorithm (i.e. running in time $O(n^{3-\epsilon})$ for some $\epsilon>0$) for APSP.

Fortunately, there are some existing truly subcubic algorithms for min-plus product when the structure of matrices is refined, and many of them can be transformed into subcubic algorithms for APSP in some situations. Noga Alon et al.~\cite{1997On} show that when the entries in the $n\times n$ matrices are integers with absolute value smaller than $W,$ the min-plus product can be solved in time $\tilde{O}\left(n^\omega\cdot W\right).$ Here $\omega<2.373$ is the cost of square matrix multiplication~\cite{1987Matrix, alman2020refined}. This result indicates a subcubic algorithm for min-plus product on square matrices when the elements are not larger than $O(n^{3-\omega-\epsilon}),$ for $\epsilon>0.$

Bringmann et al.~\cite{ref1} give a subcubic algorithm for min-plus product on bounded-difference matrices, which is defined as
\begin{definition}\label{ref1}
    A matrix is called $\delta$-bounded-difference if it satisfies that the difference between any two adjacent entries is bounded by $\delta.$ When $\delta=O(1),$ we call them bounded-difference matrices for brevity. Formally, A $\delta$-bounded-difference matrix $X$ satisfies that for all $ i,j, $ (when the elements are in $X$)
    \begin{equation*}
        \vert X_{i,j}-X_{i,j+1}\vert ~<~ \delta,
    \end{equation*}
    \begin{equation*}
        \vert X_{i,j}-X_{i+1,j}\vert ~<~ \delta.
    \end{equation*}
\end{definition}

Their algorithm~\cite{ref1} runs in time $\tilde{O}(n^{2.824}),$ given $\omega<2.373$~\cite{alman2020refined}. They also provide algorithm for released structure of bounded-difference matrices: either row or column of one matrix is bounded-difference.

Another significant work by Williams and Xu~\cite{2020Truly} provides an algorithm on less structured matrices. If two $n\times n$ matrices $A$ and $B$ satisfy that their entries are integers within $poly(\log{n})$ bits, and $B$ can be divided into $n^\delta\times n^\delta$ blocks with $W$-approximate rank at most $d=O(1),$ then given the rank decompositions, we can calculate the min-plus product in time $\tilde{O}\left(n^{3-\frac{\delta}{[d+1]/2}}+W^{1/4}n^{(9+\omega)/4}\right).$

\subsection{Our Result}


Let $\omega$ denote the exponent of time complexity of square matrix multiplication. For two rectangular matrices of sizes $n^a\times n^b$ and $n^b\times n^c,$ denote the time complexity for multiplying them by $\tilde{O}\left(n^{\omega(a,b,c)}\right)$. A current result on matrix multiplication gives $\omega<2.373, $ provided by Alman and V.V. Williams~\cite{alman2020refined}.

The main result of this paper is a faster algorithm for computing min-plus product on bounded-difference matrices.

\begin{theorem}
There is an algorithm to solve the min-plus product of any two bounded-difference $n\times n$ matrices, which works in randomized time $\tilde{O}\left(n^{2+\omega/3}\right)=\tilde{O}(n^{2.791})$, for $\omega<2.373$~\cite{alman2020refined}. By rectangular matrix multiplication algorithms~\cite{GU18}, the running time can be improved to $\tilde{O}(n^{2.779})$.
\end{theorem}

Our algorithm is faster than the one given by Bringmann et al.~\cite{ref1}. Another point is that the time complexity can be expressed in a clearer form on $\omega$.



Min-plus product on bounded-difference matrices has many applications in practice, and on many interdisciplinary scenarios. \cite{ref1} has shown that the scored parsing and language edit distance for context-free grammars, the RNA-folding problem in biology, and the optimum stack generation problem can all be reduced to min-plus product for bounded-difference matrices.

\subsection{Overview of our algorithms}

In this paper we first introduce a simpler version of our algorithm, which only partition the matrices once. Then we make the algorithm faster by recursively partition the matrices. Let $A,B$ be two $n\times n$ bounded-difference matrices and we want to compute $C=A\star B$. To make the overview clearer, we assume $\omega=2$.

\paragraph{The basic algorithm.}


As in~\cite{ref1}, in the beginning we divide the matrices $A,B$ into blocks of size $n^{0.1}\times n^{0.1}$. The upper-left entry in each block is considered as the representative entry in that block, and since the matrix is bounded-difference, we can approximately think that the entries in each block have similar values. We can thus calculate approximate results for all entries in $C$ in $O(n^{2.7})$ time, and what remains is to get the exact $C$. A concept we need here is the candidate set. Intuitively, from the approximate solution, we can find the candidate set $K(i,j)$ composed of blocks which may contain the entries $k$ that gives $C_{i,j}=A_{i,k}+B_{k,j}.$ We can judge whether an block is in the candidate set only by the representative elements of blocks. If the candidate set is small ($<n^{0.6}$), we just need to enumerate all of its entries to find the optimal one.  The time is $\tilde{O}(n^{2.7})$, since there are $O(n^{1.8})$ blocks in $C$ and trivial computation of each block takes $O(n^{0.3})$ time. For the remaining large candidate sets, we randomly sample $\tilde{O}(n^{0.3})$ columns in $A$, and for every sampled column, reduce all elements in $A$ by that column. Since each of the remaining candidate sets is larger than $n^{0.6}$, with high probability we have selected at least one column for every block of $C$ which is in its candidate sets . 

Inspired by~\cite{Fredman1976}, for a randomly chosen column $r$, the matrices $A^r,B^r$ are constructed by $A^r_{i,k}=A_{i,k}-A_{i,r}$ and $B^r_{k,j}=B_{k,j}-B_{r,j}$. If $r$ is in the candidate set for the block of $C_{i,j}$, the blocks $A_{i,k}^r$ and $B_{k,j}^r$ possibly leading to a correct value in $C$ should sum up to a value close to zero. We call such pairs of blocks ``diametric blocks''. Further we define segment in $A$ as a set of blocks in the same columns that contain values in a small range, and similarly diametric segments in rows of $B$ are defined. An important observation here is that based on the bounded-difference property, we can treat the relationship of diametric as an almost one-to-one map (by duplicating each entries no more than three times which does not affect the asymptotic time complexity). 

Our task now is to calculate the min-plus product of each pairs of diametric segments. We call a segment ``large'' if it contains $>n^{0.6}$ blocks, then the total number of large segments is $O(n^{1.2})$. Let $A', B'$ be two matrices of size $n\times n^{1.3}$ and $n^{1.3}\times n$ respectively (of $n^{0.9}\times n^{1.2}$ and $n^{1.2}\times n^{0.9}$ blocks) and each large segment in $A^r$ is allocated to a separate column of blocks in $A'$, and its diametric segment in $B$ is put in the corresponding row of blocks in $B'$. Then we can transform the entries of the new matrix to small values (between $-O(n^{0.1})$ and $O(n^{0.1})$), and compute the matrix product of the two matrices of polynomials. The running time is $\tilde{O}(n^{2.4})$ for each $r$, so $\tilde{O}(n^{2.7})$ in total. 

However, for small segments of $A^r$, if we randomly allocate them to $A'$, it is possible that two segments collide in the same column. We handle this by polynomials: two polynomials simply adds up in a collision. But if we calculate the matrix product of $A'$ and $B'$ with polynomial entries, we will get many terms coming from unwanted collisions. Now what remains is to deal with the collisions. However, since we randomly allocate the segments, the collisions are fairly evenly distributed in the result $C'=A'\cdot B'$, and we can see that each block in $C'$ will get $\tilde{O}(n^{0.6})$ collisions in expectation. And for every block in $C$, we only need to compute it in one of the sampled $r$, so subtracting collisions for all the $O(n^{1.8})$ blocks takes $\tilde{O}(n^{1.8+0.6+0.3})=\tilde{O}(n^{2.7})$ time, where every collision of blocks is computed in the trivial way in $O(n^{0.3})$ time. Finding collisions are also done in the trivial way.


\paragraph{The improved recursive algorithm.}

The improved recursive algorithm working in time $\tilde{O}(n^{2+2/3})$ follows the structure of the basic algorithm. However, more delicate techniques are applied. The core idea here is some kind of recursion: we can divide blocks into halves repeatedly, until each block has a good property or it is of constant size. This procedure takes no more than $O(\log{n})$ iterations. 

First we can partition the matrices into blocks of size $n^{1/3}\times n^{1/3}$, and in every iteration each block is divided into 4 sub-blocks. In every iteration if the size of the candidate set is larger than $n^{2/3}$, we compute it by the sampling procedure in the previous algorithm (for the current block size). If the size of the candidate set is no larger than $n^{2/3}$, we find the candidate set for smaller blocks, which can be done by simply enumerate all blocks in the current candidate sets, until in some iteration the candidate set becomes larger than $n^{2/3}$. If the candidate set is still small when the block size reaches constant, we can compute it by exhaustive search.

The main problem for this procedure is how to find the collisions when block size becomes smaller, since exhaustive search will cost more. We design a more careful recursive allocation method so that each time we only need to find collisions inside the collisions of larger blocks.

\subsection{Related work}
There are already many graph algorithms that can be accelerated by fast matrix multiplication. The all-pair shortest path problem for unweighted undirected graphs can be solved in $\tilde{O}(n^{\omega})$ time~\cite{Seidel95}, while the time complexity for APSP in unweighted directed graphs is $\tilde{O}(n^{\mu})$~\cite{Zwick98}, where $\mu<0.5286$ is closely related to rectangular matrix multiplication algorithms~\cite{GU18}. The dominance product has algorithm with complexity $O(n^{(3+\omega)/2})=O(n^{2.686})$~\cite{Matousek91}, and has been slightly improved by rectangular matrix multiplication~\cite{Yuster09}. The complexity of $O(n^{(3+\omega)/2})$, which is the exponential ``middle point'' between cubic time and FMM, is the current time complexity for many problems, such as all-pair bottleneck path~\cite{VWY07,DP09b}, all-pair non-decreasing path~\cite{Williams10,DJW}, approximate APSP without scaling~\cite{BKW19}. Recently, solving linear program is also shown to have current running time same as FMM~\cite{CLS19,JSWZ21}. 

\subsection{Structure of our paper}
The structure of this paper goes as follows. We first introduce a preliminary version of our algorithm in Section~\ref{section-basicsettings}. This algorithm works in randomized time $O(n^{2.7})$ if $\omega=2$ and $\tilde{O}(n^{2.8062})$ by fast rectangular matrix multiplications~\cite{gall12, GU18}. This basic algorithm contains most of the techniques we apply in our algorithms, and by itself it is fast and worth exploring. Then the ultimate version of our algorithm appears in Section~\ref{section-new}, which works in time $\tilde{O}(n^{2+\omega/3})$. 

	\section{Basic algorithm}\label{section-basicsettings}
	
	In this section, we examine the min-plus product of two $\delta$-bounded difference matrices with size $n \times n$. We denote these two matrices as $A$ and $B$, and the min-plus product of them as $C$. For brevity, we consider only the case where $n$ is an integer power of two; further, for brevity we assume all powers $n^{q}$ be integer powers of two. 
	
	As in~\cite{ref1}, we take a constant $\alpha\in(0, 1)$ and divide the matrices $A, B$ into $n^{\alpha}$ rows and $n^{\alpha}$ columns, so there are in total $n^{2\alpha}$ blocks of small matrices with size $n^{1-\alpha}\times n^{1-\alpha}$. We also divide $[n]=\{1,2,\cdots, n\}$ into intervals of length $n^{1-\alpha}$. For an index $i$, we define $I(i)$ as the interval of length $n^{1-\alpha}$ it lies in, and $i'$ be the first element in the interval $I(i)$, which is called the \emph{representative element} of the interval. As a result, for any element $A_{i,j}$ in matrix $A,$ $A_{i',j'}$ denotes the upper-left element of the small matrix block it lies in, and call it the \emph{representative element} of the small matrix block. 
	We denote $[n]'$ as the set of all representative elements in $[n]$, i.e., $[n]' = \{1, n^{1-\alpha} + 1, 2n^{1-\alpha} + 1, \cdots, n-n^{1-\alpha}+1 \}. $
	Further, note that $A_{I(i), I(j)}$ is the small matrix block where $A_{i, j}$ is located, and also $A_{I(i), I(j)}= A_{I(i'), I(j')}. $
	
	We construct the approximation matrix $\tilde C$ of the matrix $C$ as follows.
	$$\tilde C_{i, j} ~=~ \min_{k'\in [n]'} \{ A_{i', k'}+ B_{k', j'}\}. $$ 
	In fact, in this step we can only consider the representative elements of each small matrix block of $\tilde{C}$ to estimate the solution of the original problem by computing the min-plus product of two $\delta n^{1-\alpha}$-bounded difference matrices with size $n^{\alpha} \times n^{\alpha}$. The following lemma makes use of the fact that the difference of the elements in each small matrix block are relatively small, which in turn shows that $\tilde C$ is a good approximation to the matrix $C$: 
	
	\begin{lemma}\label{05_approx}
		For $ i, j, k \in [n]$,
		\begin{align*}
		|A_{i, k} - A_{i', k'}| ~\leq~& 2\delta n^{1- \alpha},\\
		|B_{k, j} - B_{k', j'}| ~\leq~& 2\delta n^{1- \alpha},\\
		|C_{i, j} - C_{i', j'}| ~\leq~& 2\delta n^{1- \alpha},\\
		|C_{i, j} - \tilde C_{i, j}| ~\leq~& 4\delta n^{1- \alpha}.
		\end{align*}
	\end{lemma}
	\begin{proof}
		The proofs of the first two assertions are symmetric. Since $A_{i, k}$ and $A_{i', k'}$ are in the same small matrix block $A_{I(i), I(k)}$, the difference between the row order of these two elements $|i - i'|$ is no more than $ n^{1-\alpha}. $ Combining the bounded-difference property, we get that $|A_{i, k} - A_{i', k}|~\leq~\delta n^{1-\alpha}. $ Similarly, we have $|A_{i', k} - A_{i', k'}|~\leq~\delta n^{1-\alpha}. $ Combining these two equations we get 
		$$|A_{i, k} - A_{i', k'}| ~\leq~ |A_{i, k} - A_{i', k}| + |A_{i', k} - A_{i', k'}| ~\leq~ 2\delta n^{1- \alpha}. $$
		
		For the two assertions on $C$, let $k_1= \arg \min_{k\in [n]} \{A_{i, k} + B_{k, j}\}$ be the index that really achieve the min-plus product $C_{i,j}$. Similarly, let $k_2= \arg \min_{k'\in [n]'} \{A_{i', k'} + B_{k', j'}\}. $ By definition we have $C_{i, j} = A_{i, k_1} + B_{k_1, j},\;\; \tilde C_{i', j'} = A_{i', k_2} + B_{k_2, j'}. $ 
		Combining the two equations we have 
		\begin{align*}
		C_{i', j'} ~\leq~& A_{i', k_1} + B_{k_1, j'} ~\leq~ A_{i, k_1} + B_{k_1, j} + 2\delta n^{1-\alpha} = C_{i, j} + 2\delta n^{1-\alpha}\\
		C_{i, j} ~\leq~& A_{i, k_2} + B_{k_2, j} ~\leq~ A_{i', k_2} + B_{k_2, j'} + 2\delta n^{1-\alpha}= \tilde C_{i', j'} + 2\delta n^{1-\alpha}\\\ 
		\tilde C_{i', j'} ~\leq~& A_{i', k'_1} + B_{k'_1, j'} ~\leq~ A_{i, k_1} + B_{k_1, j} + 4\delta n^{1-\alpha} = C_{i, j} + 4\delta n^{1- \alpha}
		\end{align*}
		Symmetrically, 
		we can prove $|C_{i, j} - C_{i', j'}| \le 2\delta n^{1- \alpha},\;\; |C_{i, j} - \tilde C_{i, j}| \le 4\delta n^{1- \alpha}. $ 
	\end{proof}
	
	\vspace{2ex}
	
	Define $K(i', j') = \{k'| A_{i', k'} + B_{k', j'} \le \tilde C_{i', j'} + 8\delta n^{1- \alpha} \}$ as the candidate set of the pair $(i',j')$.
	Corollary~\ref{chap05_cor1} below shows that the $k$ that achieves the optimal value for $C_{i,j}$ is in $K(i', j')$, and thus we only need to focus on the candidate set $K(i', j')$ for every pair $(i',j')$.
	
	\begin{corollary}\label{chap05_cor1}
		For $i, j, k \in [n],$ if $C_{i, j} = A_{i, k} + B_{k, j},$ then $k' \in K(i', j'). $
	\end{corollary}
	
	\begin{proof}
		By Lemma \ref{05_approx},
		$$A_{i', k'} + B_{k', j'} ~\leq~ A_{i, k} + B_{k, j} + 4\delta n^{1-\alpha} ~=~ C_{i, j} + 4\delta n^{1-\alpha} ~\leq~ \tilde C_{i, j} + 8\delta n^{1-\alpha}. $$
		Thus $k' \in K(i', j')$ since $\tilde{C}_{i',j'}=\tilde{C}_{i,j}.$
	\end{proof}
	
	Computing all $\tilde{C}_{i,j}$ and $K(i',j')$ can be done in the trivial way: enumerate all $i',k',j'$, which takes $O(n^{3\alpha})$ time. 
	Take a constant $\beta\in (0, 1)$. For all pairs $(i', j')$ satisfying $|K(i', j')|\le n^{\beta}$, for every $k'\in K(i',j')$, compute the min-plus product of the blocks $A_{I(i'),I(k')}$ and $B_{I(k'),I(j')}$ in the trivial way. Every min-plus product of small blocks takes $O(n^{3(1-\alpha)})$ time, and there are $O(n^{2\alpha})$ blocks of $C$, each with $O(n^{\beta})$ computation of blocks, thus the total time is bounded by $O(n^{3-\alpha+\beta}). $

	Now we need to deal with those $C_{i,j}$ satisfying $|K(i', j')|> n^{\beta}$. This part is the core of our algorithm. Inspired by~\cite{Fredman1976}, we randomly pick some $r'$s and reduce all elements in $A$ by the column $r$ of $A$ to get $A^r$, and reduce all elements in $B$ by the row $r$ of $B$ to get $B^r$. If $r\in K(i',j')$, which means $A_{i',r}+B_{r,j'}$ is an approximation of $C_{i',j'}=A_{i',k}+B_{k,j'}$, then $A^r_{i',k}+B^r_{k,j'}$ has a small absolute value, as shown below.
	\begin{lemma}\label{05_sim}
		For $ i, j, k_1, k_2\in [n], $ if~$ k_1', k_2' \in K(i', j'),$
		\begin{align*}
		|A_{i, k_1} + B_{k_1, j} - A_{i, k_2} - B_{k_2, j}| ~\le~ 16\delta n^{1- \alpha}.
		\end{align*}
	\end{lemma}
	\begin{proof}
		By Lemma~\ref{05_approx} and Corollary~\ref{chap05_cor1},
		$$C_{i, j} ~\leq~ A_{i, k_1} + B_{k_1, j} ~\leq~ A_{i', k'_1} + B_{k_1', j'} + 4\delta n^{1-\alpha} ~\leq~ \tilde C_{i', j'} + 12\delta n ^{1-\alpha} ~\leq~ C_{i, j} + 16\delta n^{1-\alpha}. $$
		Similarly, $$ A_{i, k_2} + B_{k_2, j} ~\in~ [C_{i, j},\;\; C_{i, j} + 16\delta n^{1-\alpha}]. $$ Thus the difference between them is at most the length of the interval $16\delta n^{1-\alpha}. $
	\end{proof}
	
	\vspace{2ex}
	
	
	We randomly sample $c_0\log n\cdot n^{\alpha- \beta}$ numbers among $[n]'$, where $c_0$ is a constant, and denote the sampling set by $R$. Lemma \ref{05_lem_3} shows that, with high probability, at least one number in every large candidate set is chosen in $R$. 
	
	\begin{lemma}\label{05_lem_3}
		With probability $1-n^{-\Omega(1)}$, for $ i', j'\in [n]'$, if $K(i', j') > n^{\beta},$ then
		$$K(i', j') \cap R ~\not =~ \emptyset.$$
	\end{lemma}
	\begin{proof}
		For fixed $i, j\in[n],\;l\in [c_0\log n\cdot n^{\alpha- \beta}]$, we let $x_l=1$ if the $l$-th selected value is in the set $K(i', j')$, and conversely we let $x_l = 0.$ Thus $x_1, x_2, \cdots$ are independent, identically distributed 0-1 random variables with expectation at least $n^{\beta} / n^{\alpha},$ and their sum (denoted by $X$) has expectation at least $c_0\log n \cdot n^{\alpha- \beta}\cdot n^{\beta} / n^{ \alpha} = c_0\log n. $ By Chernoff bound, for sufficiently large $n$ we have 
		$$
		\mathbb{P}\left(X< \frac{1}{c_0\log n}\cdot c_0\log n\right) ~\leq~ \exp\left\{-\left(1-\frac{1}{c_0\log n}\right)^2\cdot c_0\log n/2\right\} ~\leq~ \exp\{-c_0\log n/3\} ~=~ n^{-\Omega(c_0)}.
		$$
		
		Notice that $X < 1$ is actually equivalent to the event that $K(i', j') \cap R = \emptyset,$ i.e., no element in the candidate set of $(i', j')$ is selected in $R$.
		Thus the probability that $K(i', j') \cap R \not = \emptyset$ holds for all $i', j'\in [n]'$ is at least $1 - n^2\cdot n^{-\Omega(c_0)}$, so we can choose a large $c_0$ to make it a high probability.
	\end{proof}

	For $r\in R,$ we consider two $n\times n$ matrices $A^r, B^r$ satisfying $A^r_{i, k} = A_{i, k} - A_{i, r}, B^r_{k, j} = B_{k, j} - B_{r, j}. $ By Lemma \ref{05_sim} we have this natural corollary: 
	
	\begin{corollary}
		For $ i, j, k, r\in [n], $ if  $ k' \in K(i', j'), $ $ r \in R \cap K(i', j'), $
		$$|A^r_{i, k} + B^r_{k, j}| ~\leq~ 16\delta n^{1- \alpha}. $$
	\end{corollary}
	
	\vspace{2ex}
	
	We similarly divide $A^r, B^r$ into $n^{2\alpha}$ small matrix blocks with size $n^{1-\alpha}\times n^{1-\alpha}$. We call the blocks having small absolute sums as diametric blocks:
	
	\begin{lemma}\label{lem_9}
			Two blocks $A^r_{I(i), I(k)}$ and $B^r_{I(k), I(j)}$ are said to be \emph{diametric}, if they satisfy $|A^r_{i', k'} + B^r_{k', j'}| \le 16\delta n^{1- \alpha}$. Then we can see for $ i, j, k, r\in [n], $ if  $ k' \in K(i', j'), $ $ r \in R \cap K(i', j'), $ then $ A^r_{I(i), I(k)}$ and $B^r_{I(k), I(j)}$ are diametric. 
	\end{lemma}
	
	In Lemma \ref{05_lem_3}, such $r$ exists for each pair $(i', j')$ with high probability, and remember we only need to compute the blocks in $K(i',j')$ for every $i,j$. Thus for every $r$ we only need to examine the min-plus product of pairs of diametric small matrix blocks in $A^r$ and $B^r$ to ensure that we get the exact value of $C_{i,j}$ with large candidate set.
	
	The following operations are carried on a block.
	We use blocks to represent positional adjacency, and segments to represent numerical adjacency. For the $I(k)$-th columns of $A^r$, we segment the blocks by the values  of their representative elements $A^r_{i', k'}$: the elements ranging in $[0, 20\delta n^{1-\alpha}]$ are the first segment, the elements ranging in $(20\delta n^{1-\alpha}, 40\delta n^{1-\alpha}]$ are the second segment, and so on. We classify negative elements as segment -1, segment -2, and so on. Similarly we define segments in the $I(k)$-th rows of $B^r$, by the representative elements $B^r_{k',j'}$. Then we only need to consider the min-plus products of blocks in ``nearly diametric'' segments:
	
	
	\begin{observation}\label{ob1}
		If a block $A^r_{I(i), I(k)}$ is in segment $p$ in $I(k)$-th columns of $A^r$, then each of its diametric block $B^r_{I(k), I(j)}$ must be in segment $-p-1$, $-p$ or $-p+1$ in $I(k)$-th rows of $B^r$. 
	\end{observation}
	
	
	To deal with the nearly diametric segments, we duplicate matrices $A^r$ and $B^r$ 3 times to get $A^r_-, A^r_o, A^r_+$ and $B^r_-, B^r_o, B^r_+$, so that segment $p$ in $I(k)$-th columns in $A^r_-$ is said to be \emph{corresponding} to segment $-p-1$ in $I(k)$-th rows of $B^r_-$, and segment $p$ in $I(k)$-th columns in $A^r_o$ is said to be \emph{corresponding} to segment $-p$ in $I(k)$-th rows of $B^r_o$, and segment $p$ in $I(k)$-th columns in $A^r_+$ is said to be \emph{corresponding} to segment $-p+1$ in $I(k)$-th rows of $B^r_+$. Thus, we deal with the computation of 3 pairs: $A^r_-$ and $B^r_-$, $A^r_o$ and $B^r_o$, $A^r_+$ and $B^r_+$. For brevity, in the following we consider one of the 3 pairs, and simply write $A^r$ and $B^r$ in which a segment in $A^r$ is only corresponding to at most one segment in $B^r$.

	We deal with the segments of $A^r$ in two cases.

	\subsection{Handle large segments}\label{sec:large}
	
	Take a constant $\gamma\in (0, 1)$. If a segment of $ A^r $ has at least $ n^{\gamma} $ matrix blocks, we say that these segments are large. Since in columns $I(k)$ for some $k$ there are $n^{\alpha}$ blocks, there are at most $n^{\alpha-\gamma}$ large segments in columns $I(k)$, thus in total at most $n^{2\alpha-\gamma}$ large segments in $A^r$. The following lemma is easy to see:
		\begin{lemma}[\cite{1997On}]\label{lemma:trivial}
		Suppose $ A $ and $ B $ are two $ n\times n $ square matrices. Their entries are either $ +\infty $ or an integer between $ -M $ and $ M. $ Then $ A\star B $ can be solved in time $ \tilde{O}\left(n^\omega\cdot M\right). $
		
	\end{lemma}
	
	\begin{proof}
		In the beginning, we introduce a transformation $ A\to A(x) $ and $ B\to B(x)$. $ A(x) $ is a matrix of the same size with matrix $ A. $ For an entry $ i $ in $ A, $ its corresponding entry in $A(x)$ is set to be $ x^{i+M}. $ In particular, we allow that some of the values in $ A $ can be $ +\infty, $ whose corresponding entry in $ A(x) $ is therefore $ 0. $ $B\to B(x)$ is transformed in the same way. Then we calculate the matrix multiplication of $ C(x)=A(x)\cdot B(x), $ and the lowest degree of each entry in $ C(x) $ minus $2M$ should equal to that entry in $ C $ in a natural way. 
	\end{proof}
	
	Now since we only need to consider the min-plus product of every segment in $A^r$ with its corresponding segment in $B^r$, we can put every segment of $A^r$ in a separate column of blocks, and its corresponding segment of $B^r$ in the corresponding row of blocks. That is,
	
	\begin{itemize}
	    \item We make rectangular matrices $A_E^r$ and $B_E^r$ with sizes $n\times n^{1+\alpha-\gamma}$ and $n^{1+\alpha-\gamma}\times n$, respectively. From a block point of view, their sizes are $ n^{\alpha}\times n^{2\alpha -\gamma}$ and $n^{2\alpha -\gamma}\times n^{\alpha}$ while each entry is a block with size $ n^{1 - \alpha}\times n^{1 - \alpha}. $ This is the point of view we apply, and we will not separate the blocks in the operations.
	    \item We arbitrarily allocate each large segment in $A^r$ to a separate column of blocks $I(k)$ in $A_E^r$, and put its corresponding segment in $B^r$ to the row of blocks $I(k)$ in $B_E^r$. Empty elements are made $+\infty$. Adjust non-infinite values in $A_E^r$ and $B_E^r$ into a range within $[-20\delta n^{1-\alpha}, 20\delta n^{1-\alpha}]$ by adding the same value to the elements in columns $I(k)$ in $A_E^r$ and subtracting the same value to the corresponding rows $I(k)$ in $B_E^r$.
	    \item Using Lemma~\ref{lemma:trivial} to compute the min-plus product of $A_E^r$ and $B_E^r$. It is easy to see that we have computed the min-plus product of each large segment in $A^r$ with its corresponding segment in $B^r$.
	\end{itemize}
	
	The time needed to compute the min-plus product of $A_E^r$ and $B_E^r$ is $\tilde{O}\left(n^{\omega(1,1+\alpha-\gamma,1)}\cdot n^{1-\alpha}\right)$, so for all $r\in R$ (and each pair of the 3 duplicated matrices), the total time is $\tilde{O}\left(n^{\omega(1,1+\alpha-\gamma,1)}\cdot n^{1-\alpha}\cdot n^{\alpha-\beta}\right)=\tilde{O}\left(n^{\omega(1,1+\alpha-\gamma,1)+1-\beta}\right)$.
	
	\begin{remark}
	    We can deal with small segments (segment with fewer than $n^{\gamma}$ blocks) trivially, that is, directly compute the min-plus products of every pair of blocks from a small segment in $A^r$ and its corresponding segment in $B^r$. Suppose $\omega=2$ and $\alpha\rightarrow 1$ (blocks of constant size), and we can choose $\beta=0.75$ and $\gamma=0.5$, then large segment step and trivial small segment step both take $\tilde{O}(n^{2.5})$ time, in total $\tilde{O}(n^{2.75})$ for all $r\in R$. To get better running time, we need new ideas to deal with small segments.
	\end{remark}

	

	\subsection{Handle small segments}\label{sec:small}
	
	Still by the method in previous subsection, we randomly allocate each small segments of $A^r$ to a column of blocks in $A_F^r$ which is of the same size with $A_E^r$, but here different small segment can be allocated in the same column of blocks. When we compute the min-plus product of $A_F^r$ and $B_F^r$ by transforming to polynomials, many terms from segments which are not corresponding to each other (called ``collisions'') are added in the result, so we need to subtract those terms. However, since we ``randomly'' allocate the segments to column of blocks, the collisions are distributed in the result quite evenly, so we can bound the expected number of collisions for every block in the result. Also, remember that we only need to consider each block in one of the sampling $r\in R$, thus the total number of subtraction can be bounded.

	Recall that we are dealing with a randomly sampled $ r\in R $, and the matrix $ A $ and $ B $ has been adjusted to $ A^r $ and $ B^r. $ As in the previous subsection, what we will do first is to extend $ A^r,B^r $ to rectangular matrices, but here we consider the rectangular matrices with polynomial elements $ A_F^r(x)$ and $B_F^r(x)$. (Suppose we have already remove all large segments in $A^r$ and their corresponding segments in $B^r$.)
	
	\begin{itemize}
	    \item Let $A_F^r(x)$ and $B_F^r(x)$ to matrices of polynomials with sizes $ n\times n^{1+ \alpha -\gamma}$ and $n^{1+ \alpha -\gamma}\times n$, respectively. From a block point of view, theirs sizes are $n^{\alpha}\times n^{2\alpha -\gamma}$ and $n^{2\alpha -\gamma}\times n^{\alpha}$, while each entry is a block of size $n^{1 - \alpha}\times n^{1 - \alpha}$.
	    \item For each pair of corresponding segments, we adjust their values into a range within $ [-20\delta n^{1 - \alpha},20\delta n^{1 - \alpha}] $. 
	    Further, we transform $ A^r,B^r $ into matrices $ A^r(x),B^r(x) $ of polynomials, as what we do in Lemma~\ref{lemma:trivial}. Now the entries in $ A^r(x),B^r(x) $ are polynomials with degree no larger than $ 40\delta n^{1 - \alpha} $.
	    \item We randomly allocate each small segment in $ A^r(x) $ to a column of blocks $I(k)$ in $ A_F^r(x)$, and put its corresponding segment of $B^r(x)$ in the row of blocks $I(k)$ in $B_F^r(x)$. Since there are totally $ n^{2\alpha -\gamma} $ columns of blocks in $ A_F^r(x) $, the probability that a segment is allocated to a particular column of blocks is therefore $ n^{-2\alpha +\gamma}$. Note that two small segments can be put in the same column of blocks in $A^r_F(x)$. The point that we uses polynomial entries comes here: when two segments which are put into the same column of blocks overlap with each other, what we need to do is simply adding their polynomials up in the overlapping elements. 
	    \item Compute the matrix multiplication of $C_F^r(x)=A_F^r(x)\cdot B_F^r(x)$, which takes $\tilde{O}\left(n^{\omega(1,1+\alpha-\gamma,1)+1-\beta}\right)$ time as in Section~\ref{sec:large}. We need to subtract ``collisions'' from $C_F^r(x)$ to obtain the wanted results in $C^r$.
	\end{itemize}

	In fact, if small segments $A^r_1,A^r_2,\cdots,A^r_m$ are allocated to columns $I(k)$, while their corresponding segments $B^r_1,B^r_2,\cdots,B^r_m $ are allocated to rows $I(k)$, the result we get from $ C_F^r(x)=A_F^r(x)\cdot B_F^r(x) $ in $I(k)$ is $$ (A^r_1+A^r_2+\cdots+A^r_m)\times(B^r_1+B^r_2+\cdots+B^r_m)$$ and we need to remove all the \emph{collisions} $ A^r_iB^r_j ~(i\ne j)$ to get the value $ \sum_{i}A^r_iB^r_i $ which we want.

	Now we show how to eliminate all the terms from collisions in $ C_F^r(x)$. An important remark is that the elimination is done over all $ r \in R$, or to be exact, finding out the collisions is done over all $ r\in R $, but for each block of $C_F^r(x)$, subtracting the collision blocks is only needed for one of such $r$'s.
	
	\paragraph{Finding out the collisions.} This is done trivially. Consider a fixed column of blocks $I(k')$ in $A_F^r(x)$. Let $\mathcal{A}^r_{k'}$ be the set of all the small segments in $A^r$ that are allocated to column of blocks $I(k')$ in $A_F^r(x)$, and $\mathcal{B}^r_{k'}$ be their corresponding segments in $B^r$. We find all collisions (non-corresponding pairs) by enumerating all pairs of blocks in $\mathcal{A}^r_{k'}$ and $\mathcal{B}^r_{k'}$.
	Then the time is bounded by the product of their sizes, which is: (Here $A_p$ and $B_p$ are corresponding segments.)
	\begin{equation*}
	    s ~=~ \sum_{k'}\left(\sum_{A_p\in\mathcal{A}^r_{k'},B_q\in\mathcal{B}^r_{k'},p\neq q}\vert A_p\vert~\vert B_q\vert\right).
	\end{equation*}
	We have the following lemma:

	\begin{lemma}\label{lemma:collisions}
	The expected time for finding out the collisions for one of $r\in R$ is $E(s)=\tilde{O}\left(n^{2\alpha+\gamma}\right)$.
	\end{lemma}
	
	\begin{proof}
	All the segments in $A$ are randomly and independently allocated to one of the $n^{2\alpha-\gamma}$ blocks of columns, and the total number of blocks in all segments in $A$ is bounded by $n^{2\alpha}$.
	
	We can define an indicating variable $X(p,q)$. It takes 1 when segments $A_p$ and $B_q$ are in the columns and rows with the same order, otherwise $X(p,q)=0.$ We can rewrite $s$ as
	\begin{equation*}
	    s ~=~ \sum_{A_p}\sum_{B_q,p\neq q}\vert A_p\vert~\vert B_q\vert~ X(p,q).
	\end{equation*}
	
	For $p\neq q$, since the segments are allocated randomly and independently, the probability that $A_p,B_q$ are in the same-ordered columns and rows is just $1/n^{2\alpha-\gamma}$, so the expectation of $X(p,q)$ is $1/n^{2\alpha-\gamma}$.  As a result,
	\begin{equation*}
	\begin{split}
	   E\left( s \right)&~=~E\left(\sum_{A_p}\sum_{B_q,~p\ne q}\vert A_p\vert~\vert B_q\vert~ X(p,q)\right)\\
	   &~=~\sum_{A_p}\sum_{B_q,~p\ne q}E\left(\vert A_p\vert~\vert B_q\vert~ X(p,q)\right)\\
	   &~=~\sum_{A_p}\sum_{B_q,~p\ne q}\cdot1/n^{2\alpha-\gamma}\cdot\vert A_p\vert~\vert B_q\vert\\
	   &~<~ 1/n^{2\alpha-\gamma}\cdot\sum_{A_p}\sum_{B_q}\vert A_p\vert~\vert B_q\vert\\
	   &~=~ 1/n^{2\alpha-\gamma}\cdot\left(\sum_{A_p}\vert A_p\vert\right)\left(\sum_{B_q}\vert B_q\vert\right)\\
	   &~\le~ 1/n^{2\alpha-\gamma}\cdot n^{2\alpha}\cdot n^{2\alpha}~=~n^{2\alpha+\gamma}.\\
	\end{split}
	\end{equation*}
	\end{proof}

    So the expected total time for all $r\in R$ is $\tilde{O}\left(n^{2\alpha+\gamma}\cdot n^{\alpha-\beta}\right)$.
    
    \paragraph{Subtracting collisions.} In Lemma~\ref{lemma:collisions}, we bound the expected number of collisions, but for every block $C_{I(i'),I(j')}$, we only need to subtract the collisions for one of $r\in R\cap K(i',j')$. Since we randomly allocate the segments into columns in $A_F^r(x)$ and $B_F^r(x)$, the collisions can be seen as ``evenly'' distributed, so every block in $C_F^r(x)$ has $\tilde{O}(n^{\gamma})$ collisions in expectation.
    
    \begin{lemma}\label{lemma:average}
    The expected number of collisions for every block $(I(i'),I(
    j'))$ in $C_F^r(x)$ is $\tilde{O}(n^{\gamma})$.
    \end{lemma}
    \begin{proof}
    We want to compute the expected size of:
    \begin{equation*}
	    \{k'~|~ A_p\in\mathcal{A}^r_{k'},B_q\in\mathcal{B}^r_{k'},p\neq q, ~\text{$A_p$ contains interval $I(i')$ and $B_q$ contains interval $I(j')$}\}
	\end{equation*}
	There are $O(n^{2\alpha})$ different pairs of $A_p$ and $B_q$ ($p\neq q$) which satisfies $A_p$ contains interval $I(i')$ and $B_q$ contains interval $I(j')$, since at most one segment $A_p$ in every column of blocks can contain $I(i')$, similarly for $B_q$. The probability that each such pair are allocated into the same $k'$ is $1/n^{2\alpha-\gamma}$, so the expected number of collisions for every block is $\tilde{O}(n^{\gamma})$.
    \end{proof}

    We can compute the polynomial matrix product of the $\tilde{O}(n^{\gamma})$ collisions in the trivial way, which is $\tilde{O}(n^{3(1-\alpha)})$ for each collision. Since there are $O(n^{2\alpha})$ blocks in $C$, the total time for subtracting collisions is $\tilde{O}(n^{3-\alpha+\gamma})$.
    

	\subsection{Correctness and Time Analysis}
	\begin{theorem}\label{thm:correct}
	    This algorithm can find all $C_{i,j}$ with high probability.
	\end{theorem}
	\begin{proof}
		From Corollary~\ref{chap05_cor1}, we know that the representative element $k'$ of the correct solution $k$ for $C_{i,j}=A_{i,k}+B_{k,j}$ must be in $K(i',j')$. And by Lemma~\ref{05_lem_3}, with high probability, $R$ contains at least one element of $K(i',j')$ for every pair of $i,j$, then by Lemma~\ref{lem_9} we only need to compute the corresponding segments in $A^r$ and $B^r$ for all $r\in R$, which are dealt with in Section~\ref{sec:large} and~\ref{sec:small}.
	\end{proof}

	The total time complexity is composed of:
	\begin{enumerate}
	    \item To find candidate sets, computing the min-plus product of two $n^{\alpha}\times n^{\alpha}$ matrices trivially takes $O(n^{3\alpha})$ time.
	    \item Handling candidates sets smaller than $n^{\beta}$ by computing min-plus product of $n^{\beta}$ pairs of blocks trivially takes $O(n^{2\alpha}\cdot n^{\beta}\cdot n^{3(1-\alpha)})=O(n^{3-\alpha+\beta})$ time.
	    \item Allocating the large segments and small segments of $A^r,B^r$ to $A_E^r,B_E^r$ and $A_F^r(x), B_F^r(x)$ respectively for all $r\in R$ takes $O(n^{2})$ time.
	    \item Computing the product of the rectangular matrices in Section~\ref{sec:large} and \ref{sec:small} takes $\tilde{O}\left(n^{\omega(1,1+\alpha-\gamma,1)}\cdot n^{1-\beta}\right)$ time.
	    \item Finding collisions of blocks takes $\tilde{O}(n^{3\alpha+\gamma-\beta})$ time.
	    \item Subtracting collisions takes $\tilde{O}(n^{3-\alpha+\gamma})$ time.
	\end{enumerate}
	
	If $\omega=2$, then $\omega(1,1+\alpha-\gamma,1)=2+\alpha-\gamma$, so we take $\alpha=0.9$, $\beta=\gamma=0.6$ and give the result of $\tilde{O}(n^{2.7})$. When we use the rectangular matrix multiplication result $\omega(1,1.1,1)<2.4535$ by Le Gall and Urrutia~\cite{GU18}, let $\alpha=0.9354$ and $\beta=\gamma=0.7414$, then the running time will be $\tilde{O}(n^{2.8062})$.

\section{Improved Recursive Algorithm}\label{section-new}

	

In this section, we design faster algorithms by partitioning the matrices into finer blocks. First we also divide $A,B$ and the result $C$ into blocks of size $n^{1-\alpha}\times n^{1-\alpha}$, then in each iteration the blocks will be split into $2\times 2$ sub-blocks repeatedly, until each block only has one element. 
That is, we divide $[n]$ into intervals of length $l$, where $l$ ranges from $\{n^{1-\alpha},\frac{n^{1-\alpha}}{2},\cdots,1\}$.
For an index $i$, we define $I_l(i)$ as the interval of length $l$ it lies in, and $i_l'$ be the first element in the interval $I_l(i).$ As a result, for any element $A_{i,j}$ in matrix $A,$ $A_{i'_{l},j'_{l}}$ denotes the upper-left element of the small matrix block it lies in, and we call it the $representative$ $element$ of the small matrix block of size $l\times l$. 
We denote $[n]'_{l}$ as the set of all representative elements, i.e., $[n]'_l = \{1, l + 1, 2l + 1, \cdots, n-l+1\}. $

For $l=n^{1-\alpha},\frac{n^{1-\alpha}}{2},\cdots,1$, define the approximation matrix $\tilde C^l$ of the matrix $C$ as follows, and Lemma~\ref{05_approx} still holds.
	$$\tilde C^l_{i, j} ~=~ \min_{k'_l\in [n]'_l} \left\{ A_{i'_l, k'_l}+ B_{k'_l, j_l'}\right\}. $$ 

We define $K_l(i'_l, j'_l) = \{k'_l| A_{i'_l, k'_l} + B_{k'_l, j'_l} \le \tilde C^l_{i'_l, j'_l} + 8\delta l \}$ as the candidate set of $ (i'_l,j'_l) $. As in Section~\ref{section-basicsettings}, when $K_l(i'_l,j'_l)$ is large, we sample a set $R'_l$ of columns and consider the min-plus product w.r.t. the columns in $R'_l$. However, here when $K_l(i'_l,j'_l)$ is small we will not directly compute the blocks, but divide the blocks into $2\times 2$ sub-blocks and find $K_{l/2}(i'_{l/2},j'_{l/2})$, which can be found through the following lemma.

\begin{lemma}\label{lemma:candidate2}
For $i,j,k\in [n]$, if $k'_{l/2}\in K_{l/2}(i'_{l/2}, j'_{l/2})$, then $ k'_l\in K_l(i'_l, j'_l)$. This means that the candidate set for smaller blocks must be inside the candidate set for larger blocks.
\end{lemma}
\begin{proof}
From definition 
we know $$ |A_{i'_l, k'_l} - A_{i'_{l/2}, k'_{l/2}}| ~\leq~  2\delta\cdot|i'_l - i'_{l/2}| ~\leq~ \delta l. $$

Let $ k_1= \arg \min_{k\in [n]'_l} \{A_{i'_l, k} + B_{k, j'_l}\}, $ $ k_2= \arg \min_{k\in [n]'_{l/2}} \{A_{i'_{l/2}, k} + B_{k, j'_{l/2}}\}, $ which means $ A_{ i '_{l/2}, (k_2) '_{l/2} } + B_{ (k_2) '_{l/2}, j '_{l/2} } \leq A_{ i '_{l/2}, (k_1) '_{l/2} } + B_{ (k_1) '_{l/2}, j '_{l/2} }. $ Besides, for $ k'_{l/2}\in K_{l/2}(i'_{l/2}, j'_{l/2}), $ $ A_{ i '_{l/2}, k '_{l/2} } + B_{ k '_{l/2}, j '_{l/2} } \leq A_{ i '_{l/2}, (k_2) '_{l/2} } + B_{ (k_2) '_{l/2}, j '_{l/2} } + 4\delta l. $ As a result, 

\begin{equation*}
    \begin{split}
        & A_{ i '_{l}, k '_{l} } + B_{ k '_{l}, j '_{l} } ~\leq~ A_{ i '_{l/2}, k '_{l/2} } + B_{ k '_{l/2}, j '_{l/2} } + 2\delta l\\
        ~\leq~ & A_{ i '_{l/2}, (k_2) '_{l/2} } + B_{ (k_2) '_{l/2}, j '_{l/2} } + 6\delta l ~\leq~ A_{ i '_{l/2}, (k_1) '_{l/2} } + B_{ (k_1) '_{l/2}, j '_{l/2} } + 6\delta l\\
        ~\leq~ & A_{ i '_{l}, (k_1) '_{l} } + B_{ (k_1) '_{l}, j '_{l} } + 8\delta l\\
    \end{split}
\end{equation*}
Thus $ k'_l\in K_l(i'_l, j'_l). $
\end{proof}

Take a constant $ \beta\in (0,1)$. We can handle the candidate set $ K_l(i'_l, j'_l) > n^{\beta} $ by a similar method in Section~\ref{sec:small}, so we do not need to continue to find $K_{l/2}(i'_{l/2},j'_{l/2})$ for those $(i,j)$. The procedure is as follows. 
\begin{itemize}
    \item First set $l=n^{1-\alpha}$, for all $ (i'_l, j'_l)$ we compute $ K_l(i'_l, j'_l) $. This takes $ O(n^{3\alpha}) $ time.
    \item We repeat the following procedure for $l=n^{1-\alpha},\frac{n^{1-\alpha}}{2},\cdots,1$. For the current $l$, the $i,j$ which satisfies $ K_l(i'_l, j'_l) > n^{\beta} $ are dealt with by the method in Section~\ref{sec:new-segment}.
    
    \item For $ |K_l(i'_l, j'_l)| \leq n^{\beta} $, we continue to the next block size $l/2\times l/2$. Divide  block $ C_{I_l(i), I_l(j)} $ into 4 smaller blocks with size $ l/2 \times l/2 $ and compute the candidate set of each of these 4 blocks. 
    By Lemma~\ref{lemma:candidate2}, we can find $K_{l/2}(i'_{l/2}, j'_{l/2})$ inside $K_{l}(i'_{l}, j'_{l})$, and $|K_{l/2}(i'_{l/2}, j'_{l/2})| \leq 2|K_l(i'_l, j'_l)|.$  The time complexity is $ \tilde O (n^{\beta})$ for each block in $C$, and the total time is $\tilde{O}(n^{2+\beta})$ for all pairs $i,j$ and all $l$.
    \item Finally either the blocks are divided into elements, or its candidate set is at least $ n^\beta $. If for some pair $ (i,j) $, its candidate set is smaller than $n^{\beta}$ when $l=1$, then we just trivially enumerate the candidate sets to compute $C_{i,j}$, and the time for this step is bounded by $ \tilde O(n^{2+ \beta}) $.
\end{itemize}



To handle all the candidate sets larger than $n^{\beta}$, for all $l=n^{1-\alpha},\frac{n^{1-\alpha}}{2},\cdots,1$, denote $l=n^{1-\theta}$, so the block size is $n^{1-\theta}\times n^{1-\theta}$.
We randomly pick $c_0\log n\cdot n^{\theta- \beta}$ numbers among $[n]'_l$ to form $R'_l$.
As in Section~\ref{section-basicsettings}, for each $r\in R'_l,$ we consider two $n\times n$ matrices $A^r, B^r$ satisfying $A^r_{i, k} = A_{i, k} - A_{i, r}, B^r_{k, j} = B_{k, j} - B_{r, j}. $ We similarly divide $A^r, B^r$ into  blocks of size $ l\times l $. The notations and facts in Section~\ref{section-basicsettings} still work:

\begin{enumerate}
    \item Two matrix blocks $A^r_{I_l(i), I_l(k)},  B^r_{I_l(k), I_l(j)}$ are said to be diametric, if they satisfy $|A^r_{i'_l, k'_l} + B^r_{k'_l, j'_l}| \le 16\delta l$.
    \item Lemma \ref{05_lem_3}, \ref{lem_9} still hold, that is:
    \begin{itemize}
        \item With high probability, for $ i, j\in [n]$, if $K_l(i'_l, j'_l) > n^{\beta}$, then $K_l(i'_l, j'_l) \cap R'_l \neq \emptyset.$
        \item For $ i, j, k, r\in [n], $ if $ k'_l \in K_l(i'_l, j'_l),  r \in R'_l \cap K_l(i'_l, j'_l), $ then $ A^r_{I_l(i), I_l(k)}$ and $B^r_{I_l(k), I_l(j)}$ are diametric.
    \end{itemize}
    \item For the $I_l(k)$-th columns of $A^r$, we segment the blocks by the values of their representative elements $ A^r_{i'_l, k'_l}: $ the elements ranging in $[0, 20\delta l]$ are the first segment, the elements ranging in $(20\delta l, 40\delta l]$ are the second segment, and so on. 
    Similarly we define segments in the $I_l(k)$-th rows of $B^r$, by the representative elements $B^r_{k'_l,j'_l}$.  Observation \ref{ob1} still holds, and we only need to consider the min-plus products of blocks in corresponding segments. We duplicate matrix $A^r$ and $B^r$ three times. For brevity, in the following we consider one of the three pairs, and assume that a segment in $A^r$ is only corresponding to at most one segment in $B^r$.
\end{enumerate}

Unlike in Section~\ref{section-basicsettings}, we can deal with the large and small segments in a uniformed way, since Lemma~\ref{lemma:collisions} and Lemma~\ref{lemma:average} still holds if we allocate large segments in the same way as small segments.

\subsection{Handle corresponding segments}\label{sec:new-segment}

For all $l=n^{1-\alpha},\frac{n^{1-\alpha}}{2},\cdots,1$ and $l=n^{1-\theta}$, we have the sampled set $R'_l$ and the $K_l(i'_l,j'_l)$ for some pairs of $ (i'_l,j'_l). $ Now we want to find those $C_{i,j}$ where $K_l(i'_l,j'_l)$ is found and $|K_l(i'_l,j'_l)|>n^{\beta}$. When we adapt the algorithm in Section~\ref{sec:small} to handle segments when the block size is $l\times l$, the main issue comes from the collision-finding step. By the method in Section~\ref{sec:small}, the enumeration of all pairs of blocks in corresponding columns and rows takes $\tilde{O}(n^{3\theta-\beta+\gamma})$ time, so the time for this exhaustive search will becomes larger when blocks become smaller and $\theta$ increases.

\begin{definition}
    For any $r\in R'_l$, if $S$ is a segment in the $I_l(k)$-th columns of $A^r$ when the block size is $l\times l$, then when $A^r$ is divided into blocks of size $l/2^t\times l/2^t$, we say a segment $S'$ is a \emph{sub-segment} of $S$ if $S'$ is a segment in the $I_{l/2^t}(k)$-th columns (any subset of $I_l(k)$) of $A^r$, and $S'\subseteq S$. Also $S$ is called a \emph{parent} of $ S'. $
\end{definition}

So in our algorithm we carefully design a recursive allocation method of segments to columns so that if segments $S'$ and $T'$ are allocated in one column of blocks, then in previous iterations their parents $S$ and $T$ are also allocated in one column of larger blocks. So instead of exhaustive search for collisions, we can find collisions from the collisions of previous iterations, thus bounding the time complexity.

To allocate the segments when the block size is $l\times l$, we start from the block size $n^{1-\alpha}\times n^{1-\alpha}$, and then recursively allocation the sub-segments according to the position of their parents. We pick a parameter $\gamma(l,\theta)$ which depends on $l$ and $\theta$, but we simply write $\gamma$ for brevity. The allocation algorithm is given below: 

\begin{itemize}
    \item We first randomly allocate all the segments of $A^r$ when the block size is $n^{1-\alpha}\times n^{1-\alpha}$: Construct $A_{F, n^{1-\alpha}}^r$ and $B_{F, n^{1-\alpha}}^r$ with sizes $n^{\alpha}\times n^{2\alpha -\gamma}$ and $n^{2\alpha -\gamma}\times n^{\alpha}$, respectively, from block point of view.
   Randomly allocate each segment in $ A^r $ to a column of blocks $I_{n^{1-\alpha}}(k)$ in $ A_{F, n^{1-\alpha}}^r$, and put its corresponding segment of $B^r$ in the row of blocks $I_{n^{1-\alpha}}(k)$ in $B_{F, n^{1-\alpha}}^r$. 
   Note that two segments can be put in the same column of blocks, and they can overlap with each other, as in Section~\ref{sec:small}. 
    \item Assume we have already get $A_{F, \tilde{l}}^r$ and $B_{F,\tilde{l}}^r$ of sizes $ n^{ \tilde{\theta}}\times n^{2 \tilde{\theta} -\gamma}$ and $n^{2 \tilde{\theta} -\gamma}\times  n^{ \tilde{\theta}}$, while each entry is a block with size $n^{1 - \tilde{ \theta}}\times n^{1 -  \tilde{\theta}}$ where $\tilde{l}=n^{1-\tilde{\theta}}$. Next divide the block size to $\frac{\tilde{l}}{2}\times \frac{\tilde{l}}{2}$ and construct $ A_{F, \tilde{l}/2}^r$ and $ B_{F, \tilde{l}/2}^r$ of sizes $2n^{ \tilde{\theta}}\times 4n^{2 \tilde{\theta} -\gamma}$ and $ 4n^{2 \tilde{\theta} -\gamma} \times 2n^{ \tilde{\theta}} $ respectively, from a block point of view. For each column of blocks $I_{\tilde{l}}(k)$ in $A_{F,\tilde{l}}^r$, randomly assign 4 columns of blocks in $ A_{F, \tilde{l}/2}^r$ to it, which are disjoint from those for other columns of blocks in $A_{F,\tilde{l}}^r$. Then for each segment in $A_r$ which is allocated to $I_{\tilde{l}}(k)$ in $A_{F,\tilde{l}}^r$, randomly assign each of its sub-segments $S$ to one of the 4 columns of blocks in $ A_{F, \tilde{l}/2}^r$. Then put the corresponding segment of $S$ to the corresponding rows in $ B_{F, \tilde{l}/2}^r$.
    \item Repeat the above procedure until $ \tilde{l}= l $.
    \item Transform $A_{F, l}^r, B_{F,l}^r$ into matrices $ A_{F,l}^r(x),B_{F,l}^r(x) $ of polynomials with degree no larger than $ 40\delta l $, as in Section~\ref{sec:small}. Compute the rectangular matrix multiplication of polynomials $C_{F, l}^r(x)=A_{F, l}^r(x)\cdot B_{F, l}^r(x)$, The time needed for one matrix multiplication is $O\left(n^{\omega(1,1+\theta-\gamma,1)}\cdot n^{1-\theta}\right)$, so for all $ l $ and for all $r\in R'_l$ (and each pair of the 3 duplicated matrices), the total time is $\tilde O\left(n^{\omega(1,1+\theta-\gamma,1)}\cdot n^{1-\theta}\cdot n^{\theta-\beta}\right)=\tilde O\left(n^{\omega(1,1+\theta-\gamma,1)}\cdot n^{1-\beta}\right)$.
    
    \item As in Section~\ref{sec:small}, we need to subtract ``collisions'' from $C_{F, l}^r(x)$ to obtain the wanted results in $C^r$.
\end{itemize}

First note that although we allocate the segments in this recursive way, each final segment of block size $l\times l$ of $A^r$ is still uniformly randomly put into columns of blocks in $A_{F, l}^r$. So Lemma~\ref{lemma:average} still works, that is, the expected number of collisions for every block in $C_{F, l}^r(x)$ is $\tilde{O}(n^{\gamma})$. Thus, we need to subtract $\tilde{O}(n^{\gamma})$ collisions for all $n^{2\theta}$ blocks in $C_{F, l}^r(x)$, and trivially computing each collision takes $\tilde{O}(n^{3(1-\theta)})$ time, so the total time for subtracting collisions is $\tilde{O}(n^{3-\theta+\gamma})$.



\paragraph{Finding out the collisions.} For a collision of a block $(I_l(i'),I_l(j'))$ in $C_{F, l}^r(x)$ coming from segments $S$ in $A^r$ and $T$ in $B^r$, their parents also have collisions in previous iterations. Thus we can do the exhaustive search for the first iteration, then every time we just need to search inside the collisions of the previous iteration. The exhaustive search in Section~\ref{sec:small} for finding collisions in $A_{F, n^{1-\alpha}}^r$ and $B_{F, n^{1-\alpha}}^r$ needs $\tilde{O}(n^{2\alpha+\gamma}\cdot n^{\alpha-\beta})=\tilde{O}(n^{1+2\alpha+\gamma-\beta})$ time by Lemma~\ref{lemma:collisions}.

\begin{itemize}
    \item As mentioned before, for every $l\times l$ block $I_l(i), I_l(j)$ in $C^r$ we only need to compute its value for one of $r\in R'_l \cap K_l(i'_l,j'_l)$. For an $r\in R'_l$, denote the set of $l\times l$ blocks that needed to be computed by $\Gamma^r$, then $\sum_{r\in R'_l}|\Gamma^r|$ is $O(n^{2\theta})$. 
    \item For block size $n^{1-\alpha}\times n^{1-\alpha}$, we already have the collisions for all blocks in $C^r$ by $A_{F, n^{1-\alpha}}^r$ and $B_{F, n^{1-\alpha}}^r$. Repeat the following step for $\tilde{l}=n^{1-\alpha}/2, n^{1-\alpha}/4, n^{1-\alpha}/8,\cdots, l$:
    \item Consider all blocks of size $\tilde{l}\times \tilde{l}$ in $C^r$ containing some blocks of size $l\times l$ in $\Gamma^r$, we need to find the collisions for those blocks by the collisions found in the previous iteration. The number of such blocks is also bounded by $|\Gamma^r|$. 
    Also, every segment of block size $\tilde{l}\times \tilde{l}$ of $A^r$ is uniformly randomly put into columns of blocks in $A^r_{F,\tilde{l}}$, so Lemma~\ref{lemma:average} still holds, that is, the expected number of collisions for every block for iteration $\tilde{l}$ is $\tilde{O}(\gamma)$. Checking every such collision only takes $O(1)$ time since it is easy to find the columns of the possible collisions of block size $\tilde{l}\times \tilde{l}$ by the collisions of block size $2\tilde{l}\times 2\tilde{l}$. The total time needed for this step for all $r$ is $\tilde{O}(\sum_{r\in R'_l}|\Gamma^r|\cdot n^{\gamma})=\tilde{O}(n^{2\theta+\gamma})$. 
\end{itemize}

    \subsection{Correctness and Time analysis}
    The correctness proof is analogous to Theorem~\ref{thm:correct}. Note that for every $i,j$, and for the largest $l$ such that $K_l(i'_l,j'_l)>n^{\beta}$, the solution of $C_{i,j}$ can be found in the procedure of Section~\ref{sec:new-segment} for such $l$. If $K_l(i'_l,j'_l)\leq n^{\beta}$ for all $l$, we just exhaustively search every element in $K_1(i,j)$.
    
	Note that the numbers of iterations for $l$ and $\tilde{l}$ are both $O(\log n)$, which can be absorbed by the $\tilde{O}$ representation. The total time complexity is composed of:
	\begin{enumerate}
	    \item To find candidate sets when $l=n^{1-\alpha}$, computing the min-plus product of two $n^{\alpha}\times n^{\alpha}$ matrices trivially takes $O(n^{3\alpha})$ time. 
	    \item Handling candidates sets smaller than $n^{\beta}$ in all iterations takes $O(n^{2}\cdot n^{\beta})=O(n^{2+\beta})$ time.
	    \item Allocating the segments of $A^r,B^r$ to $A_{F, l}^r(x), B_{F, l}^r(x)$ respectively for all $l$ and $r\in R'_l$ takes $ \tilde O\left(n^{2+\theta-\beta}\right) $ time.
	    \item Computing the product of the rectangular matrices takes $\tilde{O}\left(n^{\omega(1,1+\theta-\gamma,1)}\cdot n^{1-\beta}\right)$ time, for $\alpha\leq \theta\leq 1$.
	    \item Finding collisions of blocks takes $\tilde{O}(n^{1+2\alpha+\gamma-\beta}+n^{2\theta+\gamma})$ time.
	    \item Subtracting collisions takes $\tilde{O}(n^{3-\theta+\gamma})$ time.
	\end{enumerate}
	
    If we use the simple rectangular matrix multiplication result $\omega(1,1+\theta - \gamma,1)\leq \omega + \theta - \gamma$, let $\alpha=\beta=\omega/3,$ $ \gamma = \theta + \omega/3 - 1. $ Notice that $\theta \in [\alpha, 1] $ so $\gamma \in [2\omega/3 - 1, \omega/3]$, and $\theta \geq \gamma.$  Then the running time will be $\tilde{O}(n^{2+\omega/3})=\tilde{O}(n^{2.791})$ by $\omega<2.3729$~\cite{alman2020refined}.
    
    If we use the rectangular matrix multiplication result by Le Gall and Urrutia~\cite{GU18} that $\omega(1,1.2,1)<2.5366$, set $\alpha=\beta=0.7789$, $\gamma=\theta-0.2211$, then the running time will be $\tilde{O}(n^{2.779})$.

\bibliography{biblio}
	
\end{document}